\documentclass[letterpaper, draft,12pt]{ieeetran}

\IEEEoverridecommandlockouts                              
\usepackage{amsmath,amsfonts,amssymb}
\usepackage{mathrsfs}
\newtheorem{theorem}{Theorem}[section]
\newtheorem{remark}{Remark}[section]

\newtheorem{lemma}[theorem]{Lemma}

\newtheorem{definition}{Definition}[section]

\newcommand{\eqs}[1]{\begin{equation*}#1\end{equation*}}

\renewcommand{\bar}{\overline}

\renewcommand{\to}{\rightarrow}

\renewcommand{\span}{\mathrm{span}}

\newcommand{\cB}{\mathcal{B}}

\newcommand{\cN}{\mathcal{N}}

\renewcommand{\cN}{\mathcal{N}}

\newcommand{\bbR}{\mathbb{R}}

\newcommand{\bbZ}{\mathbb{Z}}

\author{Ram Somaraju and Jochen Trumpf}
\title{Degrees of Freedom of a Communication Channel and Kolmogorov
  numbers}
\date{}
\begin{document}
\maketitle
\abstract In this note, we show that the operator theoretic concept of
Kolmogorov numbers and the number of degrees of freedom at
level $\epsilon$ of a communication channel are closely
related. Linear communication channels may be modeled using linear
compact operators on Banach or Hilbert spaces and the number of
degrees of freedom of such channels is defined to be the number of
linearly independent signals that may be communicated over this
channel, where the channel is restricted by a threshold noise
level. Kolmogorov numbers are a particular example of $s$-numbers,
which are defined over the class of bounded operators between Banach
spaces. We demonstrate that these two concepts are closely related,
namely that the Kolmogorov numbers correspond to the ``jump points''
in the function relating numbers of degrees of freedom with the noise
level $\epsilon$. We also establish a useful numerical computation
result for evaluating Kolmogorov numbers of compact operators.
  
\section{Introduction}
The number of degrees of freedom of a communication channel plays an
important role in evaluating the channel's Shannon capacity (see
e.g.~\cite{Somaraju2009b}). In a  physical communication system,
because of constraints such as finite transmission power and noise at
the receiver, only finitely many (linearly independent) signals may be
exchanged between the transmitter and receiver. This physical
intuition may be captured using the concept of number of degrees of
freedom of the communication channel. The number of degrees of freedom
of a communication channel has been used in evaluating the Shannon
capacity for several physically realistic channels\footnote{The concept of  number of degrees of freedom is used implicitly in evaluating 
the capacity of Shannon's classical bandwidth limited, additive Gaussian white noise  channel. In particular it can be shown that as the bandwidth becomes large the number of degrees of freedom at level $\epsilon$ approaches the well know constant $2WT$, for all noise levels $\epsilon$ (see Gallager~\cite[Ch. 8]{Gallager1968}).} (see
e.g.~\cite{Somaraju2009b,Shannon1948,Biglieri1998,Gallager1968}). This
concept has also been used in various other problem domains such as
multi-antenna communication~\cite{Migliore2006,Hanlen2006,Xu2006},
optics~\cite{Miller2000,Piestun2000} and electromagnetic field
sampling~\cite{Bucci1989}.  

Kolmogorov numbers are a particular example of so called $s$-number
sequences. In the theory of $s$-numbers one associates with every
bounded linear operator $T$, mapping between any two Banach spaces,
a scalar sequence $s_1(T)\geq s_2(T)\geq \ldots\geq 0$ (see
e.g.~\cite{Pietsch1980,Pietsch1987}).  
A classical example in the more restricted category of compact operators
mapping between Hilbert spaces is the sequence of singular values. 
$s$-number sequences of various
types have primarily been used to classify operators based on the behaviour
of the sequence $s_n$ as $n\to \infty$. In particular, various
interesting operator ideals have been obtained based on the behaviour of these
sequences.

In this note we establish the connection between Kolmogorov numbers
and degrees of freedom of a communication channel. The remainder of
this note is organised as follows: in Section~\ref{sec:definitions}, we recall
the definitions of degrees of freedom of a communication channel and
briefly explain the physical intuition behind the definition. We also
recall the definitions of $s$-number sequences in general and Kolmogorov
numbers in particular. We present our main results, establishing the connection
between the number of degrees of freedom and Kolmogorov numbers, in
Section~\ref{sec:maiRes} and provide conclusions in the final
Section~\ref{sec:conclusion}.  

\section{Definitions}\label{sec:definitions}
\subsection{Degrees of Freedom}
The number of degrees of freedom of a communication channel has for
example been defined in~\cite{Somaraju2009b}. The physical model used
in~\cite{Somaraju2009b} is as follows: there is a normed vector space
$X$ of transmitter functions and a normed vector space $Y$ of receiver
functions and a channel operator $T:X\to Y$ that is a compact linear
operator. Physically, the elements of $X$ can be thought of as signals
that a transmitter generates, and $T$ maps these signals onto elements
of $Y$, which is the space of signals a receiver can measure. It is
further assumed that there are two constraints on this channel: 1)
the power available for transmission is finite. This constraint is
modeled by assuming that the transmitter may only generate elements
of $X$ that have a norm less than or equal to $1$. 2) the receiver is noisy and
therefore small signals can not be measured. This constraint is
modeled by assuming that only received signals with norm greater than
some pre-specified small constant $\epsilon$ may be measured.  

It was shown in~\cite{Somaraju2009b} that the following definition may be used for the number of degrees of freedom of a compact operator that models a communication channel.
\begin{definition}[Degrees of freedom at
level~$\epsilon$]\label{Ndef:banDof} Suppose $ X$ and $ Y$
are normed spaces with norms $\|\cdot\|_{ X}$ and $\|\cdot\|_{ Y}$, respectively, and $T: X\to  Y$ is a compact operator. Also, let $\bar{B}_{r,X}(\theta)$ denote the closed ball of radius $r$ centered at $\theta$ in $X$.
Then the number of degrees of freedom  $\cN(\epsilon)$ of $T$ at level~$\epsilon$ is the
smallest\footnote{$\bbZ_0^+$ denotes the set of non-negative integers.}  $N\in \bbZ_0^+$ such that there exists a set of vectors
$\{\psi_1,\ldots,\psi_N\}\subset Y$ such that for all $x\in \bar{B}_{1,X}(0)$
\begin{eqnarray*}
    &\inf_{a_1,\ldots,a_N} \left\|Tx-\sum_{i=1}^N a_i\psi_i\right\|_{ Y}
    \leq \epsilon.&\\&&
\end{eqnarray*}
\end{definition}

It was shown in~\cite{Somaraju2009b} that the number of degrees of
freedom is well defined.
Physically, we interpret this definition as follows: if there is some constraint $\|\cdot\|_X\leq 1$ on
the space of transmitter functions\footnote{We normalised the
  transmitter power constraint by assuming that available transmission power, $P=1$. Different
  normalisations can e.g. be achieved by re-scaling the norm on $X$.} and if the receiver can only measure
signals that satisfy $\|\cdot\|_Y > \epsilon$, then the number of
degrees of freedom is the maximum number of linearly independent
signals that the receiver can distinguish under these constraints, where
the maximum is taken over all possible signal constellations.

The following theorem is a simple consequence of the above definition.
\begin{theorem}~\cite[Th. 3.2]{Somaraju2009b}\label{th:dofProp}
Suppose $ X$ and $ Y$ are normed spaces with norms $\|\cdot\|_{X}$ and $\|\cdot\|_{ Y}$, respectively, and $T: X\to Y$ is a compact operator. Let $\cN(\epsilon)$ denote the number of degrees of freedom of $T$ at level~$\epsilon$. Then
\begin{enumerate}
	\item \label{dofProp2} $\cN(\epsilon) = 0$ for all $\epsilon\geq\|T\|$.
	\item \label{dofProp5} Unless $T$ is identically zero, there
          exists an $\epsilon_0>0$ such that $\cN(\epsilon) \geq 1$
          for all $0<\epsilon<\epsilon_0$.      	
	\item \label{dofProp3} $\cN(\epsilon)$ is a non-increasing,
          upper semicontinuous function of $\epsilon$. 	
	\item \label{dofProp4} In any finite interval
          $(\epsilon_1,\epsilon_2)\subset\bbR$, with $0<\epsilon_1<
          \epsilon_2$, $\cN(\epsilon)$ has only finitely many
          discontinuities, i.e. $\cN(\epsilon)$ only takes finitely
          many non-negative integer values in any finite $\epsilon$ interval. 
\end{enumerate}
\end{theorem} 
\subsection{$s$-numbers and Kolmogorov numbers}
The axiomatic characterisation of $s$-numbers for Banach space valued
operators is due to Pietsch, cf.~\cite{Pietsch1987}. In the remainder
of this section, we assume that $X,X',Y$ and $Y'$ are Banach spaces,
$\cB(X,Y)$ denotes the set of bounded linear operators $T:X\to Y$ and
$\|\cdot\|$ is the standard operator norm on $\cB(X,Y)$. Also, we denote by $I$ the
identity operator.  
\begin{definition}
Let $s:T\mapsto (s_n(T))_{n=1}^\infty$ be a rule that assigns to every bounded operator $T$ on some Banach space a scalar sequence such that the following conditions are satisfied:
\begin{description}
\item[SN1] $\|T\| = s_1(T)\geq s_2(T)\geq \ldots\geq 0$.
\item[SN2] $s_n(S+T) \leq s_n(S) +\|T\|$, for $S,T\in \cB(X,Y)$ and $n=1,2,\ldots$.
\item[SN3] $s_n(BTA) \leq \|B\|s_n(T)\|A\|$, for $A\in \cB(X',X),T\in \cB(X,Y), B\in \cB(Y,Y')$ and $n=1,2,\ldots$. 
\item[SN4] $s_n(I) = 1$, for $n=1,2,\ldots$.
\item[SN5] $s_n(T) = 0$ if $\mathrm{rank}(T) < n$. 
\end{description}
Then $(s_n(T))_{n=1}^\infty$ is called an \emph{s-number sequence}.
\end{definition}
It turns out that in the compact, Hilbert space case, i.e. if we restrict the
above definition to the class of compact operators mapping between Hilbert
spaces, the above properties SN1--SN5 uniquely characterise the singular values of the operator (ordered in descending order). 

For any linear subspace $S$ of $Y$ let $Q_S^Y$ denote the natural surjection from $Y$ onto the quotient space $Y/S$. The numbers 
\eqs{d_n(T) \triangleq \inf\{\|Q_S^YT\|: \mathrm{dim}(S) < n\}}
for $T\in \cB(X,Y)$ define an $s$-number sequence~\cite{Pietsch1987},
and are called \emph{Kolmogorov numbers}.  

\section{Main results}\label{sec:maiRes}
Let $X$ and $Y$ be Banach spaces with norms $\|\cdot\|_X$ and
$\|\cdot\|_Y$, respectively, let $T:X\to Y$ be compact\footnote{We
  only consider compact $T$, because the number of degrees of freedom
  is only defined for compact operators. One can conclude from the
  arguments in~\cite{Somaraju2009b} that for physically realistic
  models of communication channels, the channel operator must be
  compact. However, it may be possible to generalise the theory
  presented in the current paper to bounded operators that are not
  necessarily compact.} and suppose $B(X)$ is the unit ball in
$X$. Let $\cN(\epsilon)$ denote the number of degrees of freedom of
$T$ at level $\epsilon$ and let $d_n = d_n(T)$ denote the $n^{th}$
Kolmogorov number. 

For ease of notaion we set $\cN(\epsilon) = \infty$ for $\epsilon <0$ for the remainder of this document. 
Now, consider the sequence  $\{\sigma_n =
\sigma_n(T)\},n = 1,2,\ldots$ implicitely defined as follows
\begin{eqnarray*}
  &\sup_{\epsilon > \sigma_n} \cN(\epsilon) = n-1& \textrm{and}\\
  &\inf_{\epsilon < \sigma_n} \cN(\epsilon) = N \geq n.&
\end{eqnarray*}
Further, if $n<N$ then for all $m$ such that $n < m \leq N$, $\sigma_m \triangleq
\sigma_n$. The sequence $\{\sigma_n\}$ was defined
in~\cite{Somaraju2009b} and $\sigma_n$ was called the $n^{th}$ DOF
singular value of $T$ in~\cite{Somaraju2009b}. It should be noted that
in the Hilbert space situation, $\sigma_n$ is simply the $n^{th}$
singular value of $T$. Also, $\cN(\epsilon)$ is
discontinuous exactly at $\epsilon =\sigma_n$ for $n=1,2,\ldots$. We now show that $\sigma_n$ is in fact equal to $d_n$ for which we will need the following elementary lemma.
\begin{lemma}
\eqs{d_{n+1} = \inf_{\psi_1,\ldots,\psi_n\in Y} \sup_{x\in B(X)} \inf_{a_1,\ldots,a_n\in \bbR} \left\|Tx - \sum_{i=1}^n a_i\psi_i\right\|}
\end{lemma}
\begin{proof}
Let $S$ be any subspace of $Y$ such that $\mathrm{dim}(S)\leq n$. By the definition of $Q_S^Y$, 
\eqs{\|Q_S^YT\| = \sup_{x\in B(X)}\inf_{y\in S} \|Tx - y\|_Y.}
The lemma now follows from two simple facts: for $\mathrm{dim}(S)\leq n$
there exists a set of vectors $\{\psi_1,\ldots,\psi_n\}$ such that
they span $S$, and the dimension of the span of $n$ vectors is not greater than $n$.
\end{proof}
\begin{theorem}\label{th:mainRes}
$\sigma_n = d_n.$
\end{theorem}
\begin{proof}
Because $\sigma_1 = d_1 = \|T\|$, we only need to prove that $\sigma_{n+1} = d_{n+1}$ for $n = 1,2,\ldots$. 

Now suppose $\sigma_N = \sigma_{N-1} = \ldots = \sigma_{n+1} < \sigma_n$. We first prove that $\sigma_{n+1} \geq d_{n+1}$. Let $\delta > 0$ be some arbitrary number satisfying
$\delta < \sigma_n - \sigma_{n+1}$. Then for $\epsilon = \sigma_{n+1} + \delta$, 
\begin{equation*}
\cN(\epsilon) = n.
\end{equation*} 
Therefore, there exists a set $\{\psi_1,\ldots,\psi_n\}\in Y$ such that 
\begin{equation*}
   \sup_{x\in B(X)} \inf_{a_1,\ldots,a_n} \left\|Tx-\sum_{i=1}^n a_i\psi_i\right\|_{ Y}
    \leq \epsilon = \sigma_{n+1} + \delta.
\end{equation*}
Because, the constant $\delta> 0$ can be made arbitrarily small,
\begin{eqnarray}
d_{n+1} &=& \inf_{\psi_1,\ldots,\psi_n\in Y} \sup_{x\in B(X)} \inf_{a_1,\ldots,a_n\in \bbR} \left\|Tx - \sum_{i=1}^n a_i\psi_i\right\|\nonumber\\ &\leq& \sigma_{n+1}\label{e1}.
\end{eqnarray}
Relabeling $N = n+1$, we get 
\begin{equation*}
d_N \leq \sigma_N = \sigma_{n+1}.
\end{equation*}
To prove the converse inequality, let $\epsilon < \sigma_{N}$. Then, $\cN(\epsilon) > N$. Therefore, for all sets $\{\psi_1,\ldots,\psi_N\}\subset Y$, there exists an $x\in B(X)$ such that,
\begin{equation*}
    \inf_{a_1,\ldots,a_N} \left\|Tx-\sum_{i=1}^{N} a_i\psi_i\right\|_{ Y}
    > \epsilon > \sigma_{N} = \sigma_{n+1}.
\end{equation*}
Therefore,
\begin{eqnarray}
d_{n+1} &\geq& d_{N}\nonumber\\ &=& \inf_{\psi_1,\ldots,\psi_N\in Y} \sup_{x\in B(X)} \inf_{a_1,\ldots,a_N\in \bbR} \left\|Tx - \sum_{i=1}^{N} a_i\psi_i\right\|\nonumber\\ &\geq& \sigma_{N}\nonumber\\ &=& \sigma_{n+1}\label{e2}.
\end{eqnarray}
Here, we used the $s$-number property SN1 that $d_n$ is non-increasing in $n$. The result now follows from~\eqref{e1} and~\eqref{e2}.
\end{proof}
\begin{remark} The number of degrees of freedom of a communication
  channel is defined for compact operators mapping between arbitrary normed
  spaces, while the concept of Kolmogorov numbers is defined only
  for operators mapping between Banach spaces. 
  However, if $X$ and $Y$ are normed spaces and $T:X\to Y$ is compact,
  then consider the completions $\tilde{X}$ and $\tilde{Y}$ of $X$ and
  $Y$, respectively, so that $\tilde{X}$ and $\tilde{Y}$ are Banach
  spaces.
  Let $\tilde{T}:\tilde{X}\to\tilde{Y}$ be the standard extension of
  $T$ from $X$ to $\tilde{X}$. 
  Then $\sigma_n(T) = \sigma_n(\tilde{T})$. This fact follows directly
  from the compactness of $T$ and the definition of $\cN(\cdot)$. 
\end{remark}

We also restate below a useful result in~\cite{Somaraju2009b} that aids in the numerical computation of Kolgomorov numbers for compact operators. 
\begin{theorem}~\cite[Th. 3.8]{Somaraju2009b} Suppose $X$ and $ Y$ are Banach spaces
and $T: X\to Y$ is a compact operator. Also suppose that $X$ has a
complete Schauder basis $\{\phi_1,\phi_2,\ldots\}$ and let
$S_m=\mathrm{span}\{\phi_1,\ldots,\phi_m\}$. Let $T_m = T|_{S_m}:S_m\to Y$,
$m \in \bbZ^+$. If $\sigma_n$, the $n^{th}$ DOF singular value of $T$, exists then for
$m$ large enough $\sigma_{n,m}$, the $n^{th}$ DOF singular value of
$T_m$, will exist and 
\begin{equation*}
    \lim_{m\to\infty}\sigma_{n,m} = \sigma_n.
\end{equation*}
If $\sigma_{n,m}$ exists then it is a lower bound for
$\sigma_n$.
\end{theorem}
We can therefore use finite dimensional approximations of $T$ to
numerically compute the Kolgomorov numbers.
 
\section{Conclusion}\label{sec:conclusion}
In this note we establish the connection between Kolgomorov numbers
and degrees of freedom of a communication channel. Specifically, we
show that the jump points (discontinuities) of the function relating
the number of degrees of freedom to the noise level $\epsilon$ are
equal to the Kolgomorov numbers. 
This connection invites the question as to whether other $s$-number
sequences such as Gelfand numbers or approximation numbers play an
equally important role in communication systems.

\end{document}